\documentclass{llncs}
\usepackage{amssymb,amsbsy,amsmath,amsfonts,amssymb,amscd}

\usepackage{complexity}
\usepackage{soul}
\usepackage{bbold}
\usepackage{color}
\usepackage{graphicx}
\newcommand{\EnumP}{\mathrm{\mathsf EnumP}}
\newcommand{\IncP}{\mathrm{\mathsf IncP}}
\newcommand{\DelayP}{\mathrm{\mathsf DelayP}}

\newcommand{\cS}{\ensuremath{\mathcal{S}}}
\newcommand{\cH}{\ensuremath{\mathcal{H}}}
\newcommand{\cF}{\ensuremath{\mathcal{F}}}

\newcommand{\cE}{\ensuremath{\mathcal{E}}}

\newcommand{\1}{\ensuremath{\mathbb{1}}}
\newcommand{\0}{\ensuremath{\mathbb{0}}}
\newcommand{\Cl}{\ensuremath{Cl}}
\newcommand{\EnumClo}{\textsc{EnumClosure}}
\newcommand{\Clo}{\textsc{Closure}}
\newcommand{\ExtClo}{\textsc{ExtClosure}}
\newcommand{\zero}{\ensuremath{\mathbf{0}}}
\newcommand{\one}{\ensuremath{\mathbf{1}}}

\newcommand{\Pb}[3]{\noindent\textbf{#1} \\ \textbf{Input}: #2\\ \textbf{Output}: #3 \\}

\title{Efficient enumeration of solutions produced by closure operations}
\author{Arnaud Mary \inst{1} \and Yann Strozecki \inst{2}}

\institute{Universit\'e Lyon 1 ; CNRS, UMR5558, LBBE / INRIA Grenoble Rh\^one-Alpes - ERABLE \and Universit\'e de Versailles Saint-Quentin-en-Yvelines, DAVID Laboratories}
\begin{document}
\maketitle

\begin{abstract}
 In this paper we address the problem of generating all elements obtained
 by the saturation of an initial set by some operations. More precisely, we prove that
 we can generate the closure by polymorphisms of a boolean relation with a polynomial delay. This implies for instance that we can compute with polynomial delay the closure of a family of sets by any set of ``set operations'' (e.g. by union, intersection, difference, symmetric difference$\dots$). To do so, we prove that for any set of operations $\cF$, one can decide in polynomial time whether an elements belongs to the closure by $\cF$ of a family of sets. When the relation is over a domain larger than two elements, our generic enumeration method fails for some cases, since the associated decision problem is $\NP$-hard and we provide an alternative algorithm.
\end{abstract}

\section{Introduction}

In enumeration complexity we are interested in listing a set of elements, which can be  of exponential cardinality in the size of the input. The complexity of these problems is thus measured in term of the input size and output size.
The enumeration algorithm with a complexity polynomial in both the input and output are called
output polynomial or total polynomial time.
Another, more precise notion of complexity, is the \emph{delay} which measures the time between the production of
two consecutive solutions. We are especially interested in problems solvable with a  delay polynomial
in the input size, which are considered as the \emph{tractable problems} in enumeration complexity.
For instance, the maximal independent sets of a graph can be enumerated with polynomial delay~\cite{JohnsonP88}.

If we allow the delay to grow during the algorithm, we obtain incremental delay algorithms: the first $k$ solutions can be enumerated
in a time polynomial in $k$ and in the size of the input.
Many problems which can be solved with an incremental delay have the following form:
given a set of elements and a polynomial time function acting on tuples of elements, produce the closure of the set by the function.
For instance, the best algorithm to generate all circuits of a matroid is in incremental delay because it uses
some closure property of the circuits~\cite{1113216}.

Polynomial delay algorithms are in incremental delay. In this article, we try to understand when saturation problems
which are natural incremental delay problems can be in fact solved by a polynomial delay algorithm.
To attack this question we need to restrict the saturation operation. In this article, an element will be a vector over
some finite set and we ask the saturation operation to act \emph{coefficient-wise} and in the same way on each coefficient.
We prove that, when the vector is over the boolean domain, every possible saturation can be computed in polynomial delay.
To do that we study a decision version of our problem, denoted by $\Clo_{\mathcal{F}}$: given a vector $v$ and a set of vectors $\cS$
decide whether $v$ belongs to the closure of $\cS$ by the operations of $\mathcal{F}$. We prove $\Clo_{\mathcal{F}} \in \P$ for all set of operations $\mathcal{F}$ over the boolean domain.

When the domain is boolean, the problem can be reformulated in term of set systems or hypergraphs.
It is equivalent to generating the smallest hypergraph which contains a given hypergraph and which is closed by some operation.
We show how to efficiently compute the closure of an hypergraph by any family of set operations (any operation that is the composition of unions, intersections and complementations) on the hyperedges. This extends known methods such as the closure of a hypergraph by union, by union and intersection or the generation of the cycles of a graph by computing the closure of the fundamental cycles by symmetric difference.

The closure computation is also related to constraint satisfaction problems (CSP). Indeed,
the set of vectors can be seen as a relation $R$ and the problem of generating its closure by some operation $f$
is equivalent to the computation of the smallest relation $R'$ containing $R$ such that $f$ is a polymorphism of $R'$.
There are several works on enumeration in the context of CSP, which deal with enumerating solutions of a CSP in polynomial delay~\cite{creignou1997generating,schnoor2007enumerating,bulatov2012enumerating}.
The simplest such result~\cite{creignou1997generating} states that in the boolean case, there is a polynomial delay algorithm if and only if the constraint language is Horn, anti-Horn, bijunctive or affine.
Our work is completely unrelated to these results, since we are not interested in the solutions of CSPs but only in generating the closure of relations. However, we use tools from CSPs such as the Post's lattice~\cite{post1941two}, used by Schaefer in its seminal paper~\cite{schaefer1978complexity}, and the Baker-Pixley theorem~\cite{baker1975polynomial}.

The main theorem of this article settles the complexity of a whole family
of decision problems and implies, quite surprisingly, that the backtrack search is enough to obtain a polynomial delay algorithm
to enumerate the closure of boolean vectors. For all these enumeration problems, compared to the naive saturation algorithm, our method has a better time complexity (even from a practical point of view) and a better space complexity (polynomial rather than exponential).

Moreover these algorithms may serve as a good tool to design other enumeration algorithms. One only has to express an enumeration problem
as the closure of some sufficiently small and easy to compute set of elements and then to apply the method described in this article.
Similarly,  given a potentially large set of objects, it might be convenient to represent
this set in a compact way with only some "essential" elements. We would like that these core elements act as a basis in the sense that any other element can be found by composing one or several operations on those base elements. If
one can define such a basis with respect to a set of operations, a natural question arises: How to compute efficiently all the elements from
the basis ? The results of this paper answer this question when the considered objects are
subsets of a ground set and when the operations are "set operations".


 Finally, besides the generic enumeration algorithm, we try to give for each closure rule an algorithm with the best possible complexity. In doing so, we illustrate several classical methods used to enumerate objects such as amortized backtrack search, reverse search, Gray code $\dots$

\subsection{Our results}

In Sec. \ref{sec:preliminary}, we define enumeration complexity, our problem and the backtrack search.
In Sec. \ref{sec:boolean}, we use Post's lattice, restricted through suitable reductions between clones, to
determine the complexity of $\Clo_{\mathcal{F}}$ for all set of binary operations $\mathcal{F}$.
It turns out that there are only a few types of closure operations:
\begin{enumerate}
 \item the monotone operations, in subsection \ref{sec:monotone}
 \item the addition over $\mathbb{F}_2$, in subsection \ref{sec:algebra}
 \item the set of all operations, or almost all, in subsection \ref{sec:all}
 \item an infinite hierarchy of increasingly hard to enumerate closures, related to the majority function, in subsection \ref{sec:threshold}
 \item the few limit cases of the previous hierarchy in subsection \ref{sec:limit}
\end{enumerate}

Finally, in Sec.~\ref{sec:largerdomains}, we give polynomial delay algorithm for three classes of closure operation over any domain
and prove that the method we use in the boolean case fails in one case.

\section{Preliminary}\label{sec:preliminary}
\subsection{Basic notations}
Given $n\in \mathbb{N}$, $[n]$ denotes the set $\{1,...,n\}$.
For a set $D$ and a vector $v\in D^{n}$, we denote by $v_i$ the $i^{\text{th}}$ coordinate of $v$. Let $i,j\in [n]$, we denote by $v_{i,j}$ the vector $(v_i,v_j)$. More generally,  for a subset $I=\{i_1,...,i_k\}$ of $[n]$ with $i_1<...<i_k$ we denote by $v_I$ the vector $(v_{i_1},..., v_{i_k})$. Let $\cS$ be a set of vector we denote by $\cS_I$ the set $\{ v_I \mid v \in \cS\}$.
The characteristic vector $v$ of a subset $E$ of $[n]$, is the vector in $\{0,1\}^{n}$ such that $v_i=1$ if and only if $i\in X$.

\subsection{Complexity}

In this section, we recall basic definitions about enumeration problems and their
 complexity, for further details and examples see~\cite{phd_strozecki}.

 Let $\Sigma$ be some finite alphabet.
 An \emph{enumeration problem} is a function $A$ from $\Sigma^*$ to $\mathcal{P}(\Sigma^*)$.
 That is to each input word, $A$ associates a set of words. An algorithm which solves
 the enumeration problem $A$ takes any input word $w$ and produces the set $A(w)$ word by word and \emph{without redundancies}.
 We always require the sets $A(w)$ to be finite. We may also ask $A(w)$ to contain only words of polynomial size in the size of $w$ and
  that one can test whether an element belongs to $A(w)$  in polynomial time. If those two conditions hold,
  the problem is in the class $\EnumP$ which is the counterpart of $\NP$ for enumeration.
   Because of this relationship to $\NP$, we often call solutions the elements we enumerate.

The computational model is the random access machine model (RAM) with addition, subtraction
and multiplication as its basic arithmetic operations. We have additional output registers,
and when a special OUTPUT instruction is executed, the contents of the output registers is outputted and considered as an element of the outputted set. We choose the RAM model, because it is closer to real computers and because it allows to store
all found solutions and to look for one in logarithmic time in the number of solutions.

The \emph{delay} is the time between the productions of two consecutive solutions.
Usually we want to bound the delay of an algorithm for all pairs of consecutive solutions
and for all inputs of the same size. If this delay is polynomial in the size of the input, then
we say that the algorithm is in \emph{polynomial delay} and the problem is in the class $\DelayP$.
If the delay is polynomial in the input and the number of already generated solutions, we say that the algorithm is in
\emph{incremental delay} and the problem is in the class $\IncP$. By definition we have $\DelayP \subset \IncP$.
Moreover $(\DelayP \cap \EnumP) \neq (\IncP \cap \EnumP)$ modulo the exponential time hypothesis~\cite{strozecki2015note}.
In practice problems in $\DelayP$ are much more tractable, often because they can be solved with a memory polynomial in
the size of the input. Note that in a polynomial delay algorithm we allow a polynomial precomputation step,
usually to set up data structures, which is not taken into account in the delay. This is why we can have a delay smaller than the size of the input.
%

We now explain a very classical and natural enumeration method called the \emph{Backtrack Search} (sometimes also called the \emph{flashlight method})
used in many previous articles~\cite{read1975bounds,strozecki2013}.
We represent the solutions we want to enumerate as vectors of size $n$ and coefficients in $[d]$. In practice solutions are often
subsets of $[n]$ which means that $k = 1$ and the vector is the characteristic vector of the subset.

The  enumeration algorithm is a depth first traversal of a tree whose nodes are partial solutions.
The nodes of the tree will be all vectors $v$ of size $l$, for all $l \leq n$, such that
$v = w_{[l]}$ and $w$ is a solution. The children of the node $v$ will be
the vectors of size $l+1$, which restricted to $[l]$ are equal to $v$.
The leaves of this tree are the solutions of our problem, therefore a depth first traversal will visit all
leaves and yield all solutions. We want an enumeration algorithm with a delay polynomial in $n$.
Since a branch of the tree is of size $n$, we need to be able to find the children of a node in a time polynomial in $n$
to obtain a polynomial delay. The delay also depends linearly on $d$, but in the rest of the paper $d$ will be constant.
Therefore the problem is reduced to the following \emph{decision} problem:
given $v$ of size $l$ is there $w$ a solution such that $v = w_{[l]}$ ?
This problem is called the \emph{extension problem} associated to the enumeration problem.

\begin{proposition}\label{prop:partialsolutions}
 Given an enumeration problem $A$, such that for all $w$, $A(w)$ can be seen as vectors of size $n$ and coefficients in $[d]$,
 with $n$ and $d$ polynomially related to $|w|$.
 If the extension problem associated to $A$ is in $\P$, then $A$ is in $\DelayP$.
\end{proposition}

More precisely, the delay is $n$ times the complexity of solving the extension problem times $d$.
We will see in the next part, that the complexity of solving the extension problem can be amortized
over a whole branch of the tree, since we solve it many times, using well chosen data structures.

There is a second enumeration method to design a polynomial delay algorithm, named the \emph{supergraph} method~\cite{avis1996reverse}.
The idea is to organize the solutions (and not the partial solutions) as a DAG instead of a tree, and to traverse this DAG.
For that we should be able to visit all the successors of a node in polynomial time. To avoid to enumerate several times a node of the DAG, the  \emph{reverse search} method is often used~\cite{avis1996reverse}. It consists on defining a canonical parent computable in polynomial time for each node. This method may have a better delay than the backtrack search, because the traversal goes over solutions only.

\subsection{Closure of families by set operations}

In this subsection, we define the family of enumeration problems we want to solve.
We fix $D$ a finite domain.
Given a $t$-ary operation $f$  (a function from $D^t$ to $D$), $f$ can be naturally extended to a $t$-ary operation over vectors of the same size. For a $t$-uples of vectors of size $n$ $v^{1},\dots v^{t}$, $f$ will then acts coefficient-wise, that is for all $i\leq n$, $f(v^{1},\dots, v^{t})_i = f(v^{1}_i,\dots, v^{t}_i)$.

\begin{definition}
Let $\mathcal{F}$ be a finite set of operations over $D$.
Let $\cS$ be a set of vectors of size $n$ over $D$.
Let $\mathcal{F}^i(\cS) = \{f(v_1,\dots, v_t) \mid v_1,\dots, v_t \in \cF^{i-1}(S) \text{ and } f\in \mathcal{F} \}$ and $\mathcal{F}^0(\cS) = \cS$.
The closure of $\cS$ by $\mathcal{F}$ is $\Cl_\mathcal{F}(\cS) = \displaystyle{\cup_i \mathcal{F}^i(\cS)}$.
\end{definition}

Remark that $\Cl_{\mathcal{F}}(\cS)$ is also the smallest set which contains $\cS$ and which is closed by the operations of $\mathcal{F}$.
The set $\Cl_\mathcal{F}(\cS)$ is invariant under the operations of $\mathcal{F}$: these operations are called \emph{polymorphisms} of the set $\Cl_\mathcal{F}(\cS)$, a notion which comes from universal algebra.

As an illustration, assume that $D = \{0,1\}$ and that  $\mathcal{F} = \{ \vee \}$.
Then the elements of $\cS$ can be seen as subsets of $[n]$ (each vector of size $n$ is the characteristic vector of a subset of $[n]$) and $\Clo_{\{\vee\}}(\cS)$ is the closure by union of all
sets in $S$. Let $\cS = \{ \{1,2,4\},  \{2,3\},  \{1,3\} \}$ then $\Cl_{\{\vee\}}(\cS)= \{ \{1,2,4\},  \{1,2,3,4\}, \{2,3\},  \{1,3\},  \{1,2,3\}\}$.
Remark that $\Cl_{\{\vee\}}(\cS)$ is indeed closed by union, that is $\vee$ is a polymorphism of  $\Cl_{\{\vee\}}(\cS)$.

We now introduce the family of enumeration problems, parametrized by $\mathcal{F}$ a set of operations over $D$,
which we will try to solve in this article.

\medskip
\Pb{$\EnumClo_\mathcal{F}$}{A set of vectors $\cS$}{$\Cl_\mathcal{F}(S)$}

In all this article, we will denote the size of the vectors of $\cS$ by $n$ and the cardinal of $\cS$
by $m$. We introduce two related decision problems. First, the extension problem associated to a set of operations $\mathcal{F}$, is the problem $\ExtClo_\mathcal{F}$:
given $\cS$ a set of vectors of size $n$,  and a vector $v$ of size $l \leq n$, is there a vector $v' \in \Cl_\mathcal{F}(\cS)$ such that $v_{[l]} = v'$.
Second, the closure problem, denoted by $\Clo_\mathcal{F}$, is a restricted version of the extension problem where $v$ is of size $n$.

\begin{proposition}
 If $\Clo_\mathcal{F} \in \P$ then $\EnumClo_\mathcal{F} \in \DelayP$.
\end{proposition}

\begin{proof}
$\ExtClo_\mathcal{F}$ can be reduced to $\Clo_\mathcal{F}$.
Indeed, given a vector $v$ of size $l$, because the operations of $\mathcal{F}$ act coordinate-wise,
the two following predicates are equivalent:
\begin{itemize}
 \item  $\exists v' \in \Cl_\mathcal{F}(\cS)$ such that $v_{[l]} = v'$
 \item $v \in  \Cl_\mathcal{F}(\cS_{[l]})$
\end{itemize}
Therefore if $\Clo_\mathcal{F} \in \P$ then we have also  $\ExtClo_\mathcal{F} \in \P$.
 We can use Prop.~\ref{prop:partialsolutions} to conclude.
\qed \end{proof}

We have introduced an infinite family of problems, whose complexity we want to determine.
Several families of operations may always produce the same closure. To deal with that, we need to introduce the notion of functional clone.

\begin{definition}
 Let $\mathcal{F}$ be a finite set of operations over $D$, the functional clone generated by
 $\mathcal{F}$, denoted by $<\mathcal{F}>$, is the set of operations obtained by any composition of the operations of $\mathcal{F}$ and of the projections $\pi_k^n : D^n \to D$
 defined by $\pi_k^n(x_1,\dots,x_n) = x_k$.
\end{definition}

This notion is interesting, because two sets of functions which generate the same clone applied to the same set produce the same closure.

\begin{lemma}
For all set of operations $\mathcal{F}$ and all set of vectors $\cS$, $\Cl_\mathcal{F}(\cS) = \Cl_{<\mathcal{F}>}(\cS)$.
\end{lemma}

The number of clones over $D$ is infinite even when $D$ is the boolean domain.
However, in this case the clones form a countable lattice, called the Post's lattice~\cite{post1941two}.
Moreover there is a \emph{finite} number of well described clones plus a few very regular infinite family of clones.

\section{The Boolean Domain}\label{sec:boolean}

In this part we will prove our main theorem on the complexity of $\Clo_\mathcal{F}$, when the domain is boolean (of size $2$).
An instance of one such problem, denoted  by $\cS$, will be indifferently seen as a set of vectors of size $n$ or a set of subsets of $[n]$.

\begin{theorem}
 Let $\mathcal{F}$ be any fixed finite set of operations over the boolean domain, then $\Clo_\mathcal{F} \in \P$ and $\EnumClo_\mathcal{F} \in \DelayP$.
\end{theorem}

There is also a uniform version of the decision problem, where $\mathcal{F}$ is given as input.
It turns out that this problem is $\NP$-hard as proven in Section \ref{sec:threshold}.

To prove our main theorem, we will prove that $\Clo_\mathcal{F} \in \P$,
for each  clone $\mathcal{F}$ of the Post's lattice. We first show that for some $\mathcal{F}$
the problem $\Clo_\mathcal{F}$ can be reduced to $\Clo_\mathcal{G}$ where $\mathcal{G}$ is another clone obtained from $\mathcal{F}$.
This helps to reduce the number of cases we need to consider.

To an operation $f$ we can associate its dual $\overline{f}$ defined by $\overline{f}(s_1,\dots,s_t) = \neg{f(\neg{s_1},\dots,\neg{s_t})}$. If $\mathcal{F}$ is a set of operations, $\overline{\mathcal{F}}$ is the set of duals of operation in $\mathcal{F}$.
We denote by $\zero$ and $\one$ the constant functions which always return $0$ and $1$. By a slight abuse of notation,
we will also denote by $\zero$ the all zero vector and by $\one$ the all one vector.

\begin{proposition}\label{prop:postsimpl}
The following problems can be polynomially reduced to $\Clo_{\mathcal{F}}$:
\begin{enumerate}
 \item $\Clo_{\mathcal{F} \cup \{\zero\}}$, $\Clo_{\mathcal{F} \cup \{\one\}}$, $\Clo_{\mathcal{F} \cup \{\zero,\one\}}$
 \item $\Clo_{\overline{\mathcal{F}}} $
 \item $\Clo_{\mathcal{F}\cup \{\neg \}} $ when $\mathcal{F} = \overline{\mathcal{F}}$
\end{enumerate}
\end{proposition}
\begin{proof}
The reductions follow easily from these observations:
\begin{enumerate}
 \item $\Cl_{\mathcal{F}\cup \{f\}}(\cS) = \Cl_{\mathcal{F}}(\cS \cup \{f\})$ for $f=\zero$ or $f=\one$ and  $\cS \neq \emptyset$.
 \item $\Cl_{\overline{\mathcal{F}}}(\cS) =  \Cl_{\mathcal{F}}(\overline{\cS}) $ where $\overline{\cS}$ is the set of negation of vectors in $\cS$.
 \item $\Cl_{\mathcal{F}\cup \{\neg \}}(\cS) = \Cl_{\mathcal{F}}(\cS \cup \overline{\cS}) $  since for every $f \in \mathcal{F}$, there exists $g\in \cF $ such that $\neg(f(v_1,\dots,v_t) = \overline{f}(\neg{v_1},\dots,\neg{v_t})=g(\neg{v_1},\dots,\neg{v_t})$.
\end{enumerate}
\qed \end{proof}

In the following picture, we represent the clones which cannot be reduced to another one using Prop.~\ref{prop:postsimpl}
and that we will investigate in this article.
For a a modern presentation of all boolean clones, their bases and the Post's lattice see~\cite{reith2003optimal}.

\begin{figure}[h]
\begin{center}
\begin{tabular}{c c}

\begin{tabular}{|l|l|}\hline
Clone & Base\\\hline
$I_2$ & $\emptyset$\\\hline
$L_2$ & $x+y+z$\\\hline
$L_0$ & $x+y$\\\hline
$E_2$ & $\wedge$\\\hline
$S_{10}$ & $x \wedge( y \vee z) $ \\\hline
$S^k_{10}$ & $Th_k^{k+1}, x \wedge( y \vee z) $\\\hline
$S_{12}$ & $x \wedge (y \to z)$\\\hline
$S^k_{12}$ & $Th_k^{k+1}, x \wedge (y \to z)$\\\hline
$D_2$ & $maj$\\\hline
$D_1$ & $maj, x+y+z$\\\hline
$M_2$ & $\vee,\wedge$ \\\hline
$R_2$ & $x\,?\,y\,:\,z$\\\hline
$R_0$ & $\vee, +$\\\hline
 \end{tabular}& \hspace{1cm}
 \parbox{5cm}{
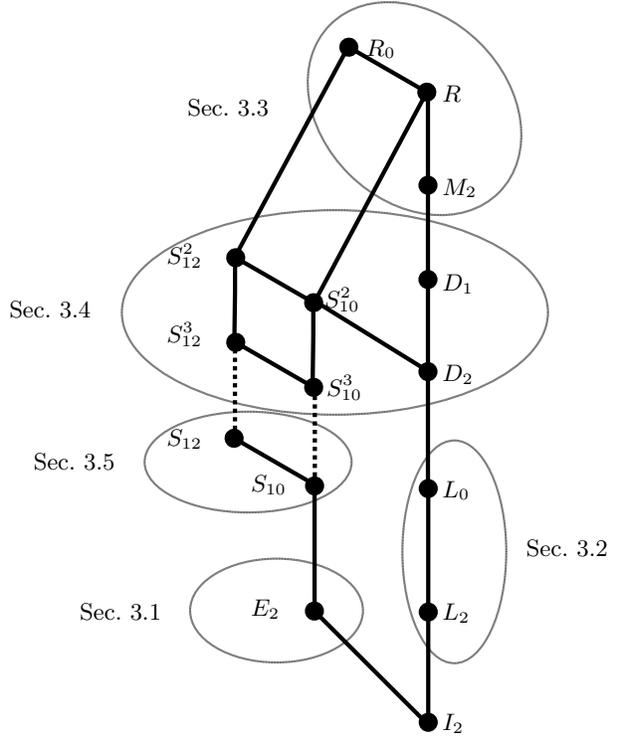
}

\end{tabular}
 \caption{The reduced Post's lattice, the edges represent inclusions of clones}
\end{center}
\end{figure}

\subsection{Conjunction}\label{sec:monotone}

We first study one of the simplest clone: $E_2 = <\wedge>$.
We give an elementary proof that $\Clo_{E_2} \in \P$, then we explain how to
obtain a good delay for $\EnumClo_{E_2}$.
For a binary vector $v$, let us denote by $\0 (v)$ (resp.  $\1 (v)$) the  set of indices $i$ for which $v_i=0$ (resp. $v_i=1$).

\begin{proposition}
$\Clo_{E_2} \in \P$.
\end{proposition}
\begin{proof}
 Let $\cS$ be a set of boolean vectors, if we apply $\wedge$ to a couple of vectors in $\cS$ it produces the intersection of two vectors when seen as sets.
 Since the intersection operation is associative and commutative, $\Cl_{E_2}(\cS)$ is the set of arbitrary intersections of elements of  $\cS$.
 Let $v$ be a vector and let $\cS_1$ be the set  $\{w \in \cS \mid w_{\1 (v)} = \one \} $.
 Assume now that $v$ can be obtained as an intersection of elements $v_1,\dots,v_t$, those elements must be in $\cS_1$ because of the monotonicity of the intersection for the inclusion.
 On the other hand, by definition of $\cS_1$, $v$ will always be smaller or equal to  $\displaystyle{\cap_{w \in \cS_1} w}$.
  Therefore, $v \in \Cl_{E_2}(\cS)$ if and only if $v = \displaystyle{\cap_{w \in \cS_1} w}$.
 This intersection can be computed in time $O(mn)$ which concludes the proof.
\qed \end{proof}

By Prop.~\ref{prop:partialsolutions}, we can turn the algorithm for $\Clo_{E_2}$ into an enumeration algorithm
for $\EnumClo_{E_2}$ with delay $O(mn^2)$. We show in the next proposition how to reduce this delay to $O(mn)$, which is the best known complexity
for this problem.

\begin{proposition}\label{prop:unionfast}
There is an algorithm solving $\EnumClo_{E_2}$ with a delay $O(mn)$.
\end{proposition}
\begin{proof}
 We use the backtrack search described in Prop. \ref{prop:partialsolutions} but we maintain data structures which allows to decide $\Clo_{E_2}$ quickly.
 Let $\cS$ be the input set of $m$ vectors of size $n$.
 During the traversal of the tree we update the partial solution $p$, represented by an array of size $n$ which stores
 whether $p_i = 1$, $p_i=0$  or is yet undefined.

 A vector $v$ of $\cS$ is compatible with the partial solution if $\1_p \subseteq \1_v$.
 We maintain an array $COMP$ indexed by the sets of $\cS$, which stores whether each vector of $\cS$ is compatible or not with the current partial solution.
 Finally we update an array $COUNT$, such that $COUNT[i]$ is  the number of compatible vectors $v \in \cS$ such that $v_i =0$.
 Remark that a partial solution $p$ can be extended into a vector of $\Cl_{E_2}(\cS)$ if and only if for all $i \in \0_p$ $COUNT[i]>0$, the solution
 is then the intersection of all compatible vectors.

 At each step of the traversal, we select an index $i$ such that $p_i$ is undefined and we set first $p_i= 0$ then  $p_i=1$.
 When we set $p_i=0$, there is no change to do in $COUNT$ and $COMP$ and we can check whether this extended partial solution is correct
 by checking if $COUNT[i]>0$ in constant time.
 When we set $p_i=1$, we need to update $COMP$ by removing from it every vector $v$ such that $v_i =0$.
 Each time we remove such a vector $v$, we decrement $COUNT[j]$ for all $j$ such that $v_j = 0$.
 If there is a $j$ such that  $COUNT[j]$ is decremented to $0$ then the extension of $p$ by $p_i = 1$ is not possible.

 When we traverse a whole branch of the tree of partial solutions during the backtrack search,
 we will set $p_i=1$ for each $i$ at most once and then we need to remove each vector from $COMP$ at most once.
 Therefore the total number of operations we do to maintain $COMP$ and $COUNT$ is $O(mn)$ and so is the delay.
\qed \end{proof}

Remark that if we ask for all extensions of the sets instead of all intersections, we exactly get the problem of enumerating the solutions
of a monotone DNF formula. In fact the algorithm used here is exactly the same as the best one to generate the solutions of a DNF formula.
Moreover, we can reduce the problem of enumerating the solutions of a monotone DNF formula to $\EnumClo_{D_2}$.
The reduction we use is called \emph{parsimonious reduction} and is relevant for counting and enumeration complexity (see~\cite{phd_strozecki}).
It maps an instance of a problem to one of another problem in polynomial time and there is a bijection between the
solutions associated to those instances which can be computed in polynomial time. For instance, since $\Cl_{\overline{D_2}}(\cS) =  \Cl_{\vee}(\overline{\cS})$
the problem of generating intersections of hyperedges reduces to the problem of generating unions of hyperedges.

\begin{proposition}
 There is a parsimonious reduction from monotone DNF to $\EnumClo_{\{\vee\}}$.
 \end{proposition}

 \begin{proof}
 Let $\phi \equiv \bigvee_{i=1}^{m} C_i$ where the $C_i$ are clauses over the variables $x_1,\dots,x_n$.
We build an hypergraph $H$ over the domain $[n]$. For each clause $C_i$, let $e_i$ be the hyperedge
$\{ i \mid x_i \in C_i\}$ and we also add the hyperedges $e_i \cup \{x_j\}$ for all $x_j \notin E$.
There is a bijection between the union of hyperedges and the solutions of the formula $\phi$.
\qed \end{proof}

Since the reduction is parsimonious, the problems of counting the elements of $\Cl_{\vee}$ and $\Cl_{D_2}$ are $\sharp \P$-hard,
while their enumerations are easy. Determining the exact complexity of $\Clo_{\vee}$ is an intriguing \emph{open problem}:
is it possible to design an algorithm with a complexity sublinear in $m$ or even which depends on $n$ only ?
Even when the input hypergraph is a graph (every set in $\cS$ is of size $2$), $m$ is bounded by $n^2$
and the question of solving $\Clo_{\vee}$ with a delay better than $O(n^3)$ is open.

\subsection{Algebraic operations}\label{sec:algebra}

We first deal with the clone $L_0 = <+>$ where $+$ is the boolean addition.
Note that $\Cl_{L_0}(\cS)$ is the vector space generated by the vectors in $\cS$.
Seen as an operation on sets, it is the symmetric difference of the two sets.

\begin{proposition}
 $\Clo_{L_0} \in \P$.
\end{proposition}
\begin{proof}
Let $\cS$ be the set of input vectors, let $v$ be a vector and let $A$ be the matrix whose rows are the elements of $\cS$.
The vector $v$ is in $\Cl_{L_0}(\cS)$ if and only if there is a solution to $Ax =v$.
Solving a linear system over $\mathbb{F}_2$ can be done in polynomial time which proves the proposition.
\qed \end{proof}

The previous proposition yields a polynomial delay algorithm by applying Prop.~\ref{prop:partialsolutions}.
One can get a better delay, by computing in polynomial time a maximal free family $M$ of $\cS$, which is
a basis of $\Cl_{L_0}(\cS)$. The basis $M$ is a succinct representation of $\Cl_{L_0}(\cS)$. One can generate all elements of $\Cl_{L_0}(\cS)$
by going over all possible subsets of elements of $M$ and summing them. The subsets can be enumerated in constant time by using Gray code enumeration (see~\cite{knuth2011combinatorial}).
The sum can be done in time $n$ by adding a single vector since two consecutive sets differ by a single element in the Gray code order.
Therefore we have, after the polynomial time computation of $M$, an enumeration in delay $O(n)$.
If one allows to output the elements represented in the basis $M$, the algorithm even has constant delay.

With some care, we can extend this result to the clone $L_2$ generated by the sum modulo two of three elements.
\begin{proposition}
 $\Clo_{L_2} \in \P$.
\end{proposition}
\begin{proof}
 First remark that any vector in $\Cl_{L_2}(\cS)$ is the sum of an odd number of vectors in $\cS$.
 In other words $v \in \Cl_{L_2}(\cS)$ if and only if there is a $x$ such that $Ax = v$ and that the Hamming weight
 of $x$ is odd. One can compute a basis $B$ of the vector space of the solutions to the equation $Ax = v$.
 If all elements of $B$ have Hamming weight even, then their sums also have Hamming weight even. Therefore
 $v \in \Cl_{L_2}(\cS)$ if and only if there is an element in $B$ with odd Hamming weight, which can be decided in polynomial time.
\qed \end{proof}

\subsection{Conjunction and disjunction}\label{sec:all}

In this subsection, we deal with the largest possible clones of our reduced Post lattice: $M_2 = <\wedge,\vee>$, $R_2 = <x \,?\, y \,:\, z>$ and $R_0= <\vee,+> $.

\begin{proposition}\label{prop:}
$\Clo_{M_2} \in \P$.
\end{proposition}
\begin{proof}
Let $\cS$ be a vector set and for all $i\in [n]$, let $X_i:=\{v\in \cS \mid v_i=1\}$. We will show that a vector $u$ belongs to $\Cl_{M_2}(\cS)$ if and only if $u:=\bigvee\limits_{i\in \1(u)} \bigwedge\limits_{v\in X_i} v$. Clearly, if $u:=\bigvee\limits_{i\in \1(u)} \bigwedge\limits_{v\in X_i} v$ then $u\in \Cl_{M_2}(\cS)$.

Assume first that there exists $i\in \1(u)$ such that $X_i=\emptyset$ i.e. for all $v\in \cS$, $v_i=0$. Then clearly, for all $w\in \Cl_{M_2}(\cS)$, $w_i=0$ and then $u\notin \Cl_{M_2}(\cS)$. Assume now that $X_i\neq \emptyset$ for all $i\in \1(u)$ and assume that $u\neq t:=\bigvee\limits_{i\in \1(u)} \bigwedge\limits_{v\in X_i} v$. So there exists $j\in \0(u)$ such that $t_j=1$. Thus, there exists $i\in \1(u)$ such that for all $v\in X_i$, $v_j=1$. We have that for all $v\in \cS$, $v_i=1\Longrightarrow v_j=1$. Let us show that this property is preserved by both operations $\wedge$ and $\vee$ and then that this property holds for all $w\in \Cl_{M_2}(\cS)$. Assume that the property holds for a set $\cF$. Let $a,b \in \cF$ and let $v:=a\wedge b$. If $v_i=1$, we have $a_i=1$ and $b_i=1$ and then $a_j=1$ and $b_j=1$. We conclude that $v_j=a_j\wedge b_j=1$. Assume now that $v=a\vee b$ and that $v_i=1$. Then either $a_i=1$ or $b_i=1$, say w.l.o.g. that $a_i=1$. Then $a_j=1$ and we have $v_j=a_j\vee b_j=1$.
We have shown that the property is preserved by both operations, therefore $u$ cannot belong to $\Cl_{M_2}(\cS)$ since $u_i=1$ and $u_j=0$.
\qed \end{proof}

When we examine the previous proof, we see that the complexity of deciding $\Clo_{M_2}$ is $O(mn^2)$ therefore by applying Prop.~\ref{prop:partialsolutions}, we get an enumeration algorithm with delay $O(mn^3)$.
We can precompute the $n$ vectors $x^i = \bigwedge_{v\in X_i} v$ and generate their unions in delay $O(n^2)$ thanks to Prop.~\ref{prop:unionfast}.
By an hill climbing algorithm, using the inclusion structure of the $x^i$ we can obtain a $O(n)$ delay.

\begin{proposition}
 $\EnumClo_{M_2}$ can be solved with delay $O(n)$.
\end{proposition}
\begin{proof}
Let $\cS$ be the input. We first build the $x^i = \bigwedge\limits_{v \in \cS, v_i = 1} v$.
The inclusion is a partial order between the $x_i$, we extend it into some total order $T$ by topological sorting.
 We then generate all elements of $\Cl_{M_2}(S)$ by an Hill climbing algorithm: we go from one solution
 to another by adding a single $x^i$. Let $v$ be the current solution, we maintain a list $L$ of the indices $i$
 of $v$ such that $v_i = 0$. At each step we select $i$ the first element of $L$ and we set $v_j =1$ and remove $j$ from $L$ for all
 $j \in \1(x^i)$. This produces a new solution in time $O(n)$. We then recursively call the algorithm on this new solution and list.
 When the recursive call is finished, we call the algorithm on $v$ and $L \setminus \{i\}$.

 This algorithm is correct, because the solutions generated in the two recursive calls are disjoint.
 Indeed, in the second call $v_i$ will always be $0$, because all indices in $L$ are smaller than $i$ in $T$.
 It means that $x^j$ for $j \in L$ is either smaller or incomparable. Since $x^i$ is the smallest element with $x^i_i =1$
 it implies that $x^j_i = 0$.
\qed \end{proof}

If we consider $\EnumClo_{M_2 \cup \{\neg\}}(\cS)$, it is very easy to enumerate.
Let $X^i = \{ v \mid v \in \cS, v_i = 1 \} \cup \{ \neg v \mid v \in \cS, v_i = 0 \}$ and
let $x^i =  \bigwedge_{v\in X_i} v$. The set $\Cl_{M_2\cup \{\neg\}}(\cS)$ is in fact a boolean algebra, whose atoms are the $x^i$.
Indeed, either $x^i_{i,j}= x^j_{i,j}$ and they are equal or $\1_{x^i} \cap \1_{x^j} = \emptyset$.
Let $A = \{ x^i \mid i\in [n]\}$, two distinct unions of elements in $A$ produce distinct elements.
Hence by enumerating all possible subsets of $A$ with a Gray code, we can generate $\Cl_{M_2 \cup \{\neg\}}(\cS)$ with a delay $O(n)$ (even $O(1)$
when always equal coefficients are grouped together).

The closures by the clones $R_2$ and $R_0$ are equal to the closure by $M_2 \cup \{\neg\}$ up to
some coefficients which are fixed to $0$ or $1$, thus they are as easy to enumerate.

\begin{proposition}
 The problems $\Clo_{R_2}$, $\Clo_{R_0}$ can be reduced to $\Clo_{M_2}$ in polynomial time.
\end{proposition}
\begin{proof}
 Let $\cS$ be a set of binary vectors. If for some $i$,  for all $v \in \cS$, $v_i=0$ (resp. $1$)
 then for all $w \in \Cl_{R_2}(\cS)$, $w_i = 0$ (resp. $1$). Therefore, we can assume that for all $i$, there is $u$ and $v$ in $\cS$
 such that $u_i = 0$ and $v_i = 1$. Remark that $ x\, ?\, x \,:\, y = x \vee y$ thus by the previous assumption we can generate $\one$.
 Let $w^l = w^{l-1} \,?\, u^l \,:\, w^{l-1}$. By assumption, we can chose $u^l$ such that $u^l_l=0$. We set $w_0=u_0$ and by a trivial induction $w_n = \zero$.
 Now remark that $ x \,?\, \zero \,:\, \one = \overline x$. Therefore we have $\Cl_{R_2}(\cS) = \Cl_{<\vee,\neg>}$ and the problem $\Clo_{<\vee,\neg>}$ can be polynomially reduced to $\Clo_{M_2}$ by point $3$ of Prop.~\ref{prop:postsimpl}.

 If for some $i$,  for all $v \in \cS$, $v_i=0$ then for all $w \in \Cl_{R_0}(\cS)$, $w_i = 0$. Therefore we can assume that for all $i$,
  there is $u\in \cS$ such that $u_i =1$. Therefore, $\one \in \Cl_{R_0}(\cS)$ by doing the union of the elements  $u^l$ such that $u^l_l = 1$.
  Finally, $x + \one = \overline{x}$ therefore we also have $\Cl_{R_0}(\cS) = \Cl_{<\vee,\neg>}$.
\qed \end{proof}

\subsection{Majority and threshold}\label{sec:threshold}

An operation $f$ is a \emph{near unanimity} of arity $k$ if it satisfies $f(x_1,x_2,\dots,x_k) = x$ for each $k$-tuple
with at most one element different from $x$. The \emph{threshold} function of arity $k$, denoted by $Th^{k}_{k-1}$, is defined by $Th^{k}_{k-1}(x_1,\dots,x_k)$ is equal to $1$ if and only if at least $k-1$ of the elements $x_1,\dots,x_k$ are equal to one.
It is the smallest near unanimity operation over the booleans.
The threshold function $Th^{3}_2$ is the majority operation over three booleans that we denote by $maj$ and the clone it generates is $D_2$.
We first give a characterization of $\Cl_{D_2}(\cS)$ which helps prove that $\Clo_{D_2} \in P$.
The characterization is a particular case of a universal algebra theorem that we then use to compute the closure by any clone which contains a threshold function.

\begin{lemma}\label{lemma:maj}
Let $\cS$ be a vector set, a vector $v$ belongs to $\Cl_{D_2}(\cS)$ if and only if for all $i,j\in [n]$, $i\neq j$, there exists $x\in \cS$ such that $x_{i,j}=v_{i,j}$.
\end{lemma}

\begin{proof}

    \noindent ($\Longrightarrow$)
     Given $a,b\in \{0,1\}$ and $i,j\in [n]$, $i\neq j$, we first show that if for all $v\in\cS$, $v_i\neq a$ or $v_j\neq b$ then for all $u\in \Cl_{D_2}(\cS)$, $v_i\neq a$ or $v_j\neq b$. It is sufficient to prove that this property is preserved by applying $maj$ to a vector set i.e. that if $\cS$ has this property, then $maj(\cS)$ has also this property. Let $x,y,z\in \cS$, $v:=maj(x,y,z)$, and assume for contradiction that $v_{i,j}=(a,b)$. Since $v_i=a$, there is at least two vectors among $\{x,y,z\}$ that are equal to $a$ at index $i$. Without loss of generality, let $x$ and $y$ be these two vectors. Since for all $u\in\cS$, $u_i\neq a$ or $u_j\neq b$, we have $x_j\neq b$ and $y_j\neq b$ and then $v_j\neq b$ which contradicts the assumption. We conclude that if $v\in \Cl_{D_2}(\cS)$, then for all $i,j\in [n]$, there exists $u\in \cS$ with $v_{i,j}=u_{i,j}$.

\noindent ($\Longleftarrow$)
  Let $k\leq n$ and let $a_1,...,a_k\in \{0,1\}$.
We will show by induction on $k$, that if for all $i,j\leq k$ there exists $v\in\cS$ with $v_i=a_i$ and $v_j=a_j$, then there exists $u\in \Cl_{D_2}(\cS)$ with $u_1=a_1$, $u_2=a_2$, $...$, $u_k=a_k$. The assertion is true for $k=2$. Assume
it is true for $k-1$, and let $a_1,...,a_k\in \{0,1\}$. By induction hypothesis there
exists a vector $w\in \Cl_{D_2}(\cS)$ with $w_1=a_1$, $...$, $w_{k-1}=a_{k-1}$.
By hypothesis, for all $i\leq k$ there exists $v^{i}\in \cS$ with
$v^{i}_{i}=a_i$ and $v^{i}_{k}=a_k$. We then construct a sequence of vectors
$(u^i)_{i\leq k}$ as follow. We let $u^{1}=v^{1}$ and for all $1<i<k$,
$u^{i}=maj(w,u^{i-1},v^{i})$. We claim that $u:=u^{k-1}$ has the property sought
i.e. for all $i\leq k$, $u_{i}=a_{i}$. First let prove that for all $i<k$ and for
all $j\leq i$, $u^{i}_{j}=a_j$. It is true for $u_{1}$ by definition.
Assume now that the property holds for $u^{i-1}$, $i<k$. Then, by construction,
for all $j\leq i-1$, we have $u^{i}_j=a_j$ since $w_j=a_j$
and $u^{i-1}_{j}=a_j$. Furthermore, we have
$u^{i}_{i}=maj(w_i,u^{i-1}_{i},v^{i}_i)=a_i$ since $w_i=a_i$ and $v_i=a_i$. We
conclude that for all $i\leq k-1$, $u_i=u^{k-1}_{i}=a_i$.

    We claim now that for all $i<k$,  $u^{i}_k=a_k$. It is true for $u^{1}$. Assume it is true for $u^{i-1}$, $i<k$. Then we have $u^{i}_k=maj(w_k,u^{i-1}_{k},v^{i}_k)$ which is equal to $a_k$ since $u^{i-1}_{k}=a_k$ by induction and $v^{i}_k=a_k$ by definition. We then have $u_i=a_i$ for all $i\leq k$ which concludes the proof.
\qed \end{proof}

\begin{corollary}
$\Clo_{D_2}$ is polynomial.
\end{corollary}
\begin{proof}
Using Lemma~\ref{lemma:maj}, one decides whether a vector $v$ is in $\Cl_{D_2}(S)$, by considering every pair of index
$i,j$ and checking whether there is a vector $w \in S$ such that $v_i,j = w_i,j$. The complexity is in $O(mn^2)$.
\qed \end{proof}

By applying Prop~\ref{prop:partialsolutions}, we get an enumeration algorithm in delay $O(mn^3)$,
and we explain how to improve this delay in the next proposition.

\begin{proposition}\label{prop:speedup}
 $\EnumClo_{D_2}$ can be solved in delay $O(n^2)$.
\end{proposition}
\begin{proof}
We do a backtrack search and we explain how to efficiently decide $\Clo_{D_2}$ during the enumeration.
 We first precompute for each pair $(i,j)$ all values $(a,b)$ such that there exists $v\in \cS$,
 $v_{i,j} = (a,b)$. When we want to decide whether the vector $v$ of size $l$ can be extended into a solution,
 it is enough that it satisfies the condition of Lemma~\ref{lemma:maj}. Moreover, we already know that
 $v_{[l-1]}$ satisfies the condition of Lemma~\ref{lemma:maj}. Hence we only have to check that the values of $v_{i,l}$ for all $i <l$
 can be found in $\cS_{i,l}$ which can be done in time $O(l)$. The delay is the sum of the complexity of deciding $\Clo_{D_2}$ for each partial solution in a branch: $O(n^2)$.
\qed \end{proof}

It turns out that Lemma~\ref{lemma:maj} is a particular case of a general theorem of universal algebra
which applies to all near unanimity terms. However we felt it was interesting to give the lemma and its proof to
get a sense of how the following theorem is proved.

\begin{theorem}[Baker-Pixley, adapted from~\cite{baker1975polynomial}]\label{thm:BP}
Let $\mathcal{F}$ be a clone which contains a near unanimity term of arity $k$,
then $v \in \Cl_{\mathcal{F}}(\cS)$ if and only if for all set of indices $I$ of size $k-1$,
$v_I \in \Cl_{\mathcal{F}}(\cS)_I$.
\end{theorem}

This allows to settle the case of $D_1 = <maj, x+y+z>$ and of the two infinite families of clones
of our restricted lattice $S_{10}^k = <Th_k^{k+1}, x \wedge (y \vee z)> $ and $S_{12}^k = <Th_k^{k+1},x \wedge (y \rightarrow z) > $.

\begin{corollary}\label{cor:coro}
If a clone $\mathcal{F}$ contains $Th_k^{k+1}$ then $\Clo_{\mathcal{F}}$ is solvable in $O(mn^{k})$.
In particular $\Clo(S_{10}^k)$, $\Clo(S_{12}^k)$ and $\Clo(D_1)$ are in $\P$.
\end{corollary}
\begin{proof}
 Let $\cS$ bet a set of vectors and let $v$ be a vector. By Th.~\ref{thm:BP}, $v\ \in \Cl_{\mathcal{F}}(\cS)$
 if and only if for all $I$, $v_I \in \Cl_{\mathcal{F}}(\cS)_I$. First remark that $\Cl_{\mathcal{F}}(\cS)_I = \Cl_{\mathcal{F}}(\cS_I)$ because
 the functions of $\mathcal{F}$ act coefficient-wise on $\cS$. The algorithm generates for each $I$ of size $k$ the set  $\Cl_{\mathcal{F}}(\cS_I)$.
 For a given $I$, we first need to build the set $\cS_I$ in time $m$ and then the generation of $\Cl_{\mathcal{F}}(\cS_I)$ can be done in constant time.
 Indeed, we can apply the classical incremental algorithm to generate the elements in $\Cl_{\mathcal{F}}(\cS_I)$, and the cardinal of $\Cl_{\mathcal{F}}(\cS_I)$
 only depends on $k$ which is a constant. The time to generate all $\Cl_{\mathcal{F}}(\cS_I)$ is $O(mn^k)$ and then all the tests can be done in $O(n^k)$.
 \qed \end{proof}

We have proved that the complexity of any closure problem in one of our infinite families is polynomial.
Remark that we can use the method of Prop.~\ref{prop:speedup} to obtain a delay $O(n^{k})$
for enumerating the elements of a set closed by a near unanimity function of arity $k$.
Notice that we could have applied Theorem~\ref{thm:BP} to the clones of Subsection~\ref{sec:all} which all contain the $maj$ function.
However, it was relevant to deal with them separately to obtain a different algorithm with delay $O(n)$ rather than $O(n^2)$.

Notice that the complexity of $\Clo_\mathcal{F}$ is increasing with the smallest arity of a near unanimity function in $\mathcal{F}$. We should thus investigate the complexity of the uniform problem when the clone is given as input. We introduce the following restricted version which turns out to be hard.

\medskip
\Pb{ClosureTreshold}{A set $\cS$ of vectors and an integer $k$}{Yes, if the vector $\mathbf{1}\in \Cl_{S_{10}^k}(\cS) $}

\begin{theorem}\label{thm:NP-completeness}
 \textbf{ClosureTreshold} is $\coNP$-complete.
\end{theorem}

\begin{proof}

First notice that the problem is in $\coNP$ since by Theorem~\ref{thm:BP}, the answer to the problem is negative if and only if one can exhibit a subset of indices of $I$ of size $k$ such that no elements of $\cS_{I}$ is equal to $\one$.

\medskip

Let us show that the Hitting Set problem can be reduced to ClosureTreshold. Given a hypergraph
$\cH=(V,\mathcal{E})$, the Hitting set  problem asks whether there exists
a subset $X\subseteq V$ of size $k$ that intersects all the hyperedges of $\cH$.
This problem is a classical NP-complete problem~\cite{garey1979computers}. Let $\cH=(V,\cE)$ be a hypergraph and
$k$ be an integer. Let $\bar{\cH}$ be the hypergraph on $V$ whose
hyperedges are the complementary of the hyperedges of $\cH$, and let
$\cS$ be the set of characteristic vectors of the
hyperedges of $\bar{\cH}$. Then $\cH$ has a transversal of size $k$ if and
only if there is a set $I$ of indices of size $k$ such for all $v\in\cS_I $,
$v \neq \one$.
Indeed, $I$ is a hitting set of $\cH$ if for all $E\in \cE$, there exists $i\in I$ such that $i\in E$ which implies that $i\notin \overline{E}$ and then the characteristic vector $v$ of $\overline{E}$ is such that $v_i=0$.

Let us show that a set $I$ of indices of size $k$ is such that no element in $\cS_I $ is equal to $\one$ if and only if no element of $\Cl_{S_{10}^k}(\cS_I)$ is equal to $\one$. We assume that $k\geq 3$ hence $S_{10}^k = <Th_k^{k+1}>$.
Remark that if no element in $\cS_I $ is equal to $\one$, then the application
of $Th_k^{k+1}$ to $\cS_I$ preserves this property. Indeed, let consider $Th_k^{k+1}(v^1,\dots,v^{k+1})$,
each $v^i$ has a zero coefficient and since there are $k+1$ such vectors and the vectors are of size $k$, by the pigeonhole principle,
there are $i,j,l$ such that $v^i_l = v^j_l = 0$. This implies that $Th_k^{k+1}(v^1,\dots,v^{k+1})\neq \one$.

Since the other direction is straightforward, we have thus proved that there is a set $I$ of indices of size $k$ such that for all $v\in\cS_I $,
$v \neq \one$ if and only if there is a set $I$ of indices of size $k$ such that for all $v\in\Cl_{S_{10}^k}(\cS_I)$,
$v \neq \one$. By Theorem~\ref{thm:BP}, the later property is equivalent to $ \one \notin \Cl_{S_{10}^k}(\cS)$.
Therefore we have given a polynomial time reduction from Hitting set to the complementary of ClosureTreshold which proves the proposition.
\qed \end{proof}

In fact, the result is even stronger. We cannot hope to get an FPT algorithm for ClosureTreshold parametrized by $k$
since the Hitting Set problem parametrized by the size of the hitting set is $\W[2]$-complete~\cite{flum2006parameterized}.
It means that if we want to significantly improve the delay of our enumeration algorithm for the clone $S_{10}^k$, we should drop
the backtrack search since it relies on solving $\Clo_{S_{10}^k}$.

\subsection{Limits of the infinite parts}\label{sec:limit}

Here we deal with the two cases left which are the limits of the two infinite hierarchies of clones we have seen in the previous subsection.
Let begin with $S_{12} = < x \wedge (y \to z)>$.

\begin{remark}\label{rmq:at_least_one}
    Let $\cS$ be a vector set and assume that there exists a $i\in [n]$ such that for all $v\in \cS$, $v_i=1$ (resp. $v_i=0$) then for all $w\in \Cl_{S_{12}}(\cS)$ we have $w_i=1$ (resp.  $w_i=0$). Then we will assume in this section that for all $i\in [n]$ there is at least a vector $v$ in $\cS$ with $v_i=1$ and a vector $w$ with $w_i=0$.
\end{remark}

\begin{theorem}\label{thm:}
    Let $\cS$ be a vector set, a vector $v$ belongs to $\Cl_{S_{12}}(\cS)$ if and only if
\begin{itemize}
  \item there exists $w\in \cS$ such that $\1 (v) \subseteq  \1 (w)$
   \item for all $(k,i)\in \1 (v)\times \0 (v)$ there exists $w\in \cS$ with $w_{k,i}=(0,1)$ or $w_{k,i}=(1,0)$
\end{itemize}

\end{theorem}
\begin{proof}
Let us start by proving the following claim.\\
\noindent\textbf{Claim}: Let $k,i\in [n]$. Then there exists $u\in \Cl_{S_{12}}(\cS)$ such that $u_{k,i}=(1,0)$ if and only if there exists $v\in \cS$ such that $v_{k,i}=(1,0)$ or $v_{k,i}=(0,1)$.

 Assume first that there exists $v\in \cS$ such that $v_{k,i}=(0,1)$. Let $x\in \cS$ such that $x_k=1$ and $y\in \cS$ such that $y_i=0$. Without loss of generality, such vectors exist by the assumption of Remark \ref{rmq:at_least_one}. Then $u:=x \wedge (v\to y)$ has the sought property, i.e $u_{k,i}=(1,0)$.
 Assume now that for all $v\in \cS$, $v_{k,i}\neq (1,0)$ and $v_{k,i}\neq (0,1)$. We show that this property is preserved by the application of
 $x \wedge (y \to z)$. For all $v\in \cS$, $v_{k,i} = (1,1)$ or $v_{k,i}= (0,0)$. Since the function $x \wedge (y \to z)$ acts coordinate-wise on the vectors, if we consider $ w = x \wedge (y \to z)$ with $x,y,z \in \cS$ we must have $w_i = w_k$. Therefore $w_{k,i} \neq  (1,0)$ and $w_{k,i}\neq (0,1)$ which implies by induction that there is no $v$ with $v_{k,i} =(0,1)$ and $v \in \Cl_{S_{12}}(\cS)$.
%
%
We can now prove the theorem.
\medskip

   \noindent($\Longleftarrow$) We can simulate $w\wedge v$  with $w\wedge (w \to v)$.
    We will show that for all $i\in \0 (v)$ either there exists a vector $v^{i}\in \cS$ such that $\1 (v)\subseteq \1 (v^{i})$ and $v^{i}_{i}=0$ or we can construct it. Notice that it is sufficient in order to prove that $v\in \Cl_{S_{12}}(\cS)$ since we have $v=\bigwedge\limits_{i\in \0 (v)} v^{i}$. So let $i\in\0 (v)$ and assume that for all $w\in \cS$  such that $\1 (v) \subseteq \1 (w)$ we have $w_i=1$. Let $w$ be such a vector and let $\1 (v)=\{j_1,j_2,...,j_k\}$. We will construct a sequence of vector $(w^{l})_{l\leq k}$ such that for all $l\leq k$ and for all $r\leq l$, $w^{l}_{j_r}=1$ and $w^{l}_{i}=0$. Let $w^{1}$ be the vector with $w^{1}_{j_1}=1$ and  $w^{1}_{i}=0$. By the claim, such a vector exists in $\Cl_{S_{12}}(\cS)$. Now for all $l\leq k$, let us define  $w^{l}:=w\wedge (u^{l} \to w^{l-1})$ where $u^{l}$ is a vector such that $u^{l}_{j_l}=0$ and $u^{l}_{i}=1$ and there is such a vector in $\Cl_{S_{12}}(\cS)$ by the claim. Since by induction we have $w^{l-1}_i=0$, and since $u^{l}_{i}=1$, we have $(u^{l} \to w^{l-1})_{i}=0$ and thus $w^{l}_{i}=0$. Now since $u^{l}_{j_l}=0$ and $w_{j_l}=1$ we have $w^{l}_{j_l}=1$. Finally, for all $r < l$, we have $w_{j_r}$ and $w^{l-1}_{j_r}=1$. Hence $w^{l}_{j_r}=1$. We obtain that  $\1 (v)\subseteq \1 (w^{k})$ and $w^{k}_i=0$.
\medskip

\noindent($\Longrightarrow$)
Let $v\in \Cl_{S_{12}}(\cS)$.
Notice that if $v=x\wedge(y\to z)$, then $\1(v)\subseteq \1(x)$. Thus, there exists $w\in \cS$ such that $\1(v)\subseteq \1(w)$. Now, by the claim, for all $k,i\in [n]$ such that $v_{k,i}=(1,0)$ there exists $w\in \cS$ such that $w_{k,i}=(1,0) $ or $w_{k,i}=(0,1)$ which conclude the proof.

\qed \end{proof}

\begin{corollary}\label{cor:celle_la_aussi}
      $\Clo_{S_{12}}$ is polynomial.
\end{corollary}

Finally, we deal with the clone $S_{10} = <x\wedge (y\vee z)>$.
The characterization of $\Cl_{S_{10}}(\cS)$ we give is very similar to the one of $\Cl_{S_{12}}(\cS)$
and the proof works in the same way.

\begin{theorem}\label{thm:}
    Let $\cS$ be a vector set, a vector $v$ belongs to $\Cl_{S_{10}}(\cS)$ if and only if
\begin{itemize}
  \item there exists $w\in \cS$ such that $\1 (v) \subseteq  \1 (w)$
   \item for all $(k,i)\in \1 (v)\times \0 (v)$ there exists $w\in \cS$ with $w_{k,i}=(1,0)$
\end{itemize}
\end{theorem}
\begin{proof}

\noindent($\Longleftarrow$)
Assume first that  $v\in \Cl_{S_{10}}(\cS)$.
Notice that if $v=x\wedge(y \vee z)$, then $\1(v)\subseteq \1(x)$. Thus by a simple induction, there exists $w\in \cS$ such that $\1(v)\subseteq \1(w)$.

Now let $(k,i)\in \1 (v)\times \0 (v)$. Let us show that if for all $w\in \cS$ $w_{k,i}\neq (1,0)$, then $u_{k_i}\neq (1,0)$ for all $u\in \Cl_{S_{10}}(\cS)$ and then $v\notin \Cl_{S_{10}}(\cS)$. It is sufficient to show that this property is preserved by the operation $x\wedge (y\vee z)$. So let $a,b$ and $c$ be three boolean vectors such that $a_{k,i}\neq (1,0)$, $b_{k,i}\neq (1,0)$, $c_{k,i}\neq (1,0)$ and let $d=a\wedge (b\vee c)$. Assume that $d_{i}=0$. Then either $a_i=0$ or both $b_i$ and $c_i$ are $0$. If $a_i=0$ then $a_k=0$ since $a_{k,i}\neq (1,0)$ and then $d_k=a_k\wedge(b_k\vee c_k=0)=0$. Now if $b_i=0$ and $c_i=0$, we have $b_k=0$ and $c_k=0$ and then $d_k=a_k\wedge(b_k\vee c_k=0)=0$. We conclude that $d_{k,i}\neq (0,1)$

\medskip

\noindent($\Longrightarrow$)
Assume that there exists $u\in \cS$ such that $\1 (v) \subseteq  \1 (u)$ and
for all $(k,i)\in \1 (v)\times \0 (v)$ there exists $w\in \cS$ with $w_{k,i}=(1
,0)$. Notice that $\wedge\in S_{10}$ since $a\wedge b= a\wedge (b\vee
 b)$. Let $t:=\bigwedge\limits_{u\in \cS,~ \1(v)\subseteq \1(u)} u$. We have $\1(
v)\subseteq \1(t)$. Either $v=t$ or there is $i\in [n]$ for which $v_{i}=0$ and $
t_i=1$. For each such coordinate $i$, we will show how to construct a vector $
t'$ such that $t'_i=0$ and such that $\1(v)\subseteq \1(t') \subseteq \1(t)$.
Let $x:=\bigvee\limits_{u\in \cS,~ u_i=0}u$. Notice that $x_i=0$ and since for all $j\in \1(v)$ there exists $w\in \cS$ such that $w_{i,j}=(0,1)$ we have $\1(v)\subseteq \1(x)$. Now let us define $t':=t\wedge x$. It is easy to see that $t'$ satisfies the conditions sought.
To construct $t'$ we proceed as follow. Let $\{y^{1},...,y^{k}\}:=\{u\in \cS \mid u_i=0\}$. Then let us construct the following sequence of vectors  $t^{1}:=t\wedge (y^{1}\vee y^{2})$, $t^{2}:=t\wedge (t^{1}\vee y^{3})$, ..., $t^{k-1}:=t\wedge (t^{k-2}\vee y^{k})$. It is easy to see that $t'=t^{k-1}$, and then $t'\in \Cl_{S_{10}}(\cS)$. We conclude that $v\in \Cl_{S_{10}}(\cS)$. Indeed starting from  $t$, we can apply the previous procedure to set to $0$ each index $i$ for which $v_i=0$ and $t_i=1$.
\qed \end{proof}

\begin{corollary}\label{prop:}
$\Clo_{S_{10}}$ is polynomial.
\end{corollary}

\section{Larger Domains}\label{sec:largerdomains}

In this section, we try to extend some results of the boolean domain to larger domains.

\subsection{Tractable closure}

We exhibit two families of clones $\mathcal{F}$ such that $\Clo_\mathcal{F} \in \P$. As a result, we obtain a polynomial delay algorithm for $\EnumClo_\mathcal{F}$
 using the backtrack search.

The first tractable case is an extension of the clones of Subsection~\ref{sec:threshold}.
Indeed using Th.~\ref{thm:BP}, we can get an equivalent to Corollary.~\ref{cor:coro} in any domain size.

\begin{corollary}
If $\mathcal{F}$ contains a near unanimity operation, then $\Clo_\mathcal{F} \in P$.
\end{corollary}

In particular, by using the same method as in Prop.~\ref{prop:speedup} we get the following result.

\begin{proposition}
 If $\mathcal{F}$ contains a near unanimity term of arity $k$, then $\EnumClo_\mathcal{F}$ can be solved in
 delay $O(n^{k-1})$.
\end{proposition}

We could hope to increase the class of polynomial time decidable problems, by using other ideas from CSP.
For instance, we may try to prove that if a clone $\mathcal{F}$ contains a Maltsev operation (it generalizes the majority operation)
then $\Clo_\mathcal{F} \in P$.

The second tractable case is a generalization of Subsection~\ref{sec:algebra}.
We consider $\mathcal{F}$ the clone generated by the addition of two elements over $D$.
To decide $\Clo_\mathcal{F}$, we have to solve a linear system, which can also be done in polynomial time over any domain.
In fact we can further extend this result as shown in the next proposition.

\begin{proposition}
 Let $f$ be a commutative group operation over $D$, then $\Clo_{<f>} \in P$.
\end{proposition}
\begin{proof}
 We want to solve $\Clo_{<f>}$, given $\cS$ a set of vectors and $v$ a vector.
  Let $A$ be the matrix which has the elements of $\cS$ as rows.
  The vector $v$ is in $\Clo_{<f>}(\cS)$ if and only there is a vector $x$ with coefficients in $\mathbb{Z}$ such that $Ax = v$.
  This equation is not over a field so we cannot solve it directly.
 We apply a classical group theorem to the finite commutative group $(D,f)$, which states that
 $D$ is a direct sum of cyclic groups $D_1,\dots,D_t$ whose order is the power of a prime.
  The equation $Ax = v$ can be seen as a set of equations over fields: $A_ix_i = v_i$, for $i\leq t$,
  where $A_i$, $x_i$ and $v_i$ are the projection of $A$, $x$ and $v$ over $D_i$. We can easily reconstruct an $x$ which have the
  projections $x_i$ on $D_i$ by the Chinese remainder theorem. Therefore, deciding whether $v \in \Clo_{<f>}(\cS)$ is equivalent to solving a
  set of linear systems and hence is in polynomial time.
 \qed \end{proof}

 One natural generalization would be to allow the function $f$ to be non commutative.
 In that case, we conjecture that $\Clo_{<f>}$ is $\NP$-hard.

\subsection{A limit to the backtrack search}

The last case we would like to extend is the clone generated by the conjunction.
A natural generalization is to fix an order on $D$ and to study the complexity
of $\Clo_{<f>}$ with $f$ monotone.
Let $f$ be the function over $D=\{0,1,2\}$ defined by  $f(x,y) = x+y$ if $x+y \leq 2$ and
$f(x,y) = 2$ otherwise. This function is clearly monotone for the usual order.

\begin{proposition}
 $\Clo_{<f>}$ is $\NP$-complete.
\end{proposition}
\begin{proof}
 We reduce EXACT-3-COVER to $\Clo_{<f>}$. Let $\cS$ be an instance of EXACT-3-COVER, that is a set of subsets of $[n]$
 of size $3$. Clearly, $\cS$ can be seen as an instance of $\Clo_{<f>}$ and we prove that $\one \in \Clo_{<f>}(\cS)$ if and only if
 there is an exact cover of $\cS$. First remark that $f$ associative,
 therefore any element of $\Clo_{<f>}(\cS)$ can be written $f(v_1,f(v_2,f(v_3, \dots)$ with $v_i \in \cS$. It is also commutative
 therefore we can associate a unique element of $\Clo_{<f>}(\cS)$ to a multiset of elements of $\cS$ by the previous construction.
 Remark that it is never useful to have three times the same element in the multiset since $f(v,v) = f(v,f(v,v))$.
 If $v_i > 0$ then $f(v_i,v_i) =2$, therefore the vector $\one$ can only be generated by a set and not a multiset.
 Moreover a set which generates $\one$ satisfies that for all $i\leq n$ there is one and only one of its elements with a coefficient $1$ at the index $i$. Such a set is an exact cover of $\cS$, which proves the reduction.

 The problem is in $\NP$ because an element $v$ is in $\Clo_{<f>}(\cS)$, if and only if
 there is a multiset of elements of $\cS$ such that applying $f$ to its elements yields $v$. This witness is of polynomial size
 since each element is at most twice in the multiset.
\qed \end{proof}

This hardness result implies that we cannot use the backtrack search to solve the associated enumeration algorithm.
However, if we allow a space proportional to the number of solutions, we can still get a polynomial delay algorithm for associative functions,
a property satisfied by the function $f$ of the last proposition. Remark that the space used can be exponential while the backtrack search
only requires a polynomial space.

\begin{proposition}
 If $f$ is an associative function, then $\EnumClo_{<f>} \in \DelayP$.
\end{proposition}
\begin{proof}
Let $\cS$ be an instance of $\EnumClo_{<f>}$. Let $G$ be the directed graph with vertices $\Cl_{<f>}(\cS)$
and from each $v \in \Cl_{<f>}(\cS)$, there is an arc to $f(v,s)$ for all $s \in \cS$.
Since $f$ is associative, by definition of $G$, every vertex of $\Cl_{<f>}(\cS)$ is accessible from a vertex in $\cS$.
Therefore we can do a depth-first traversal of the graph $G$ to enumerate all solutions.
A step of the traversal is in polynomial time: from an element $v$ we generate its neighborhood: $f(v,s)$ for $s \in \cS$.
The computation of $f(v,s)$ is in time $O(n)$ and $|\cS| = m$. We must also test whether the solution $f(v,s)$
has already been generated. This can be done in time $O(n)$ by maintaining a self balanced search tree containing the generated solutions,
since there are at most $|D|^{n}$ solutions.
In conclusion the delay of the enumeration algorithm is in $O(mn)$ thus polynomial.
\qed \end{proof}

To obtain a polynomial space algorithm, we could try to use the reverse search method.
To do that, we want the graph $G$ to be a directed acyclic graph, which is the case if we require the function
to be monotone. The monotonicity also ensures that the depth of $G$ is at most $n(|D|-1)$.
However we also need to be able to compute for each element of $G$ a canonical ancestor
in polynomial time and it does not seem to be easy even when $f$ is monotone.
We leave the question of finding a good property of $f$ which ensures the existence of an easy to compute ancestor open for future research.

\section{Further work}
\begin{itemize}
    \item Classify the complexity of $\Clo_{\mathcal{F}}$ for domains larger than two.
    \item Find $\mathcal{F}$ such that $\EnumClo_{\mathcal{F}}$ can be solved with a polynomial delay and space but such that $\Clo{\mathcal{F}}$ is $\NP$-hard.
    \item For set family (boolean domain) can we enumerate only the minimal elements (backtrack fails even with the symmetric difference)? It would give an enumeration algorithm of the circuits of binary matroids.
    \item What if we only allow operations between elements satisfying a given property (if the intersection is not empty for instance).
    \item Can we allow the closure function to have a different action on different coefficients and still obtain the same kind of results ?
\end{itemize}

\noindent\textbf{Acknowledgements}:
Authors have been partly supported by the ANR project Aggreg and we thank the members of the project and Mamadou Kanté for interesting
discussions about enumeration. We also thank Florent Madelaine for his help with CSP and universal algebra.
%

\bibliographystyle{splncs}
\bibliography{biblio.bib}
%
%
\end{document}